\documentclass[letterpaper, 10 pt, conference]{ieeeconf}
\usepackage{cite}
\usepackage{amsmath,amssymb,amsfonts}
\usepackage{algorithmic}
\usepackage{graphicx}
\usepackage{textcomp}
\usepackage[c2 , nocomma]{optidef}
\usepackage{color}
\usepackage{mathtools}
\usepackage{lipsum}
\usepackage{eucal}
\usepackage{subcaption}
\usepackage{graphicx}
\newtheorem{theorem}{Theorem}
\newtheorem{definition}{Definition}
\newtheorem{assumption}{Assumption}
\newtheorem{remark}{Remark}
\newtheorem{lemma}{Lemma}
\DeclareMathOperator*{\argmin}{arg\,min}
\def\BibTeX{{\rm B\kern-.05em{\sc i\kern-.025em b}\kern-.08em
    T\kern-.1667em\lower.7ex\hbox{E}\kern-.125emX}}

\IEEEoverridecommandlockouts 
\begin{document}

\title{\LARGE \bf Robust Safe Control Synthesis with Disturbance Observer-Based Control Barrier Functions}

\author{Ersin Daş$^{1}$ and Richard M. Murray$^{1}$
\thanks{$^{1}$Ersin Daş and Richard M. Murray are with the California Institute of Technology, 1200 East California Boulevard, Pasadena, CA 91125 USA
        {\tt\small ersindas@caltech.edu; murray@cds.caltech.edu}}
}




\maketitle
\thispagestyle{empty}
\pagestyle{empty}

\begin{abstract}
In a complex real-time operating environment, external disturbances and uncertainties adversely affect the safety, stability, and performance of dynamical systems.
This paper presents a robust stabilizing safety-critical controller synthesis framework with control Lyapunov functions (CLFs) and control barrier functions (CBFs) in the presence of disturbance. A high-gain input observer method is adapted to estimate the time-varying unmodelled dynamics of the CBF with an error bound using the first-order time derivative of the CBF. This approach leads to an easily tunable low-order disturbance estimator structure with a design parameter as it utilizes only the CBF constraint. The estimated unknown input and associated error bound are used to ensure robust safety by formulating a CLF-CBF quadratic program. The proposed method is applicable to both relative degree one and higher relative degree CBF constraints. The efficacy of the proposed approach is demonstrated using a numerical simulations of an adaptive cruise control system and a Segway platform with an external disturbance.
\end{abstract}


\section{Introduction}
\label{sec:introduction}
Real-time safety is a necessity in many control applications, for instance, autonomous vehicles, constrained robotic systems, and spacecraft.
Therefore, provable safety-critical control of dynamical systems has drawn increasing attention in recent years.
Control barrier functions (CBFs) are a tool to handle safety constraints in the form of forward invariance of a set \cite{ames2019control}.
CBFs can be unified with stability and performance requirements, encoded by the time derivative of a control Lyapunov function (CLF), in an online quadratic program (CLF-CBF-QP) to ensure safety and control objectives simultaneously \cite{ames2014control}.
This optimization framework has been widely applied to multi-agent systems \cite{xu2017realizing}, autonomous driving \cite{he2021rule}, and wheeled robots \cite{gurriet2018towards} due to its computational efficiency. 
Although these applications guarantee optimization constraints for high-fidelity dynamical models, real-time control systems generally include unavoidable uncertainties and disturbances that might cause performance degradation, and in some cases, even lead to unsafe operations \cite{xu2015robustness}. 

To address the robustness mentioned above issue of CLF-CBF-QP, the input-to-state safe CBF (ISSf-CBF) technique, which provides a robust stabilizing safe controller via a larger forward invariant set, has been introduced in \cite{kolathaya2018input}. The infinity norm of the bounded disturbance is used to synthesize the controller directly without a model of the unknown input.
More recently, tunable ISSf-CBF (TISSf-CBF) has been proposed to reduce the conservatism of the ISSf-CBF method due to the worst-case disturbance assumption \cite{alan2021safe}. In \cite{jankovic2018robust, buch2021robust}, and \cite{garg2021robust} robust CBF approaches have been investigated for nonlinear systems with disturbance to guarantee safety.

Disturbance observer theory, a well-studied robust control tool, has been used with the ISSf-CBF method and worst-case disturbance bound to attenuate external disturbances for safety-critical control of an autonomous surface vehicle \cite{gu2021safety}. 
In \cite{zhao2020adaptive}, robust safety constraints have been enforced by an adaptive pointwise unmodeled dynamic estimation law. 
This work considers relative-degree one systems in which the first time derivative of the CBF depends on the control signal. However, this restrictive assumption is violated in several robotic systems, as most safety constraints have relative degree greater than one.

In this study, a high-gain disturbance observer-based robust CLF-CBF-QP is formulated to guarantee the safety of a disturbed nonlinear system in the presence of time-varying unknown inputs.
This disturbance observer scheme integrates the first-order time derivative of CBF and CLF with an input observer approach to estimate the unmodelled system dynamics within an exponential error bound. 
Since this method uses only the first-order CBF or CLF constraint, it presents a simple disturbance estimation framework with only one design parameter that needs to be tuned.
We then formulate a CLF-CBF-QP containing estimated disturbance and error bound-based constraints that provide a robust, safe stabilizing control input.
Moreover, the proposed safe control method is appropriate for high relative degree CBF constraints.
Finally, we demonstrate the applicability of this method using adaptive cruise control (ACC) and Segway platform simulation examples. 

The rest of this paper is organized as follows.
The preliminaries are introduced in Section II. Section III provides the disturbance observer-based robust CLF-CBF-QP scheme. Simulation results are presented in Section IV, Section V concludes the paper.

\section{Preliminaries}
Notation: The notation used in this study is fairly standard.
$\mathbb{R}$, $\mathbb{R}^+$, $\mathbb{R}^+_0$ represent the set of real, positive real and non-negative real numbers, respectively.
The Euclidean norm of a matrix is denoted by $\|\cdot\|$, and $\|\cdot\|_\infty$ represents the infinity norm.
A continuous function $\alpha : \mathbb{R}^+_0 \rightarrow \mathbb{R}^+_0$ belongs to class-$\mathcal{K}_\infty$ ($\alpha \in \mathcal{K}_\infty$) if it is strictly increasing, $\alpha(0) = 0$, $\alpha(r) \rightarrow \infty$ as $r \rightarrow \infty$, and a continuous function $\alpha : \mathbb{R} \rightarrow \mathbb{R}$ belongs to extended class-$\mathcal{K}_\infty$ ($\alpha \in \mathcal{K}_{\infty, e}$) if it is strictly monotonically increasing, $\alpha(0) = 0$, $\alpha(r) \rightarrow \infty$ as $r \rightarrow \infty$, $\alpha(r) \rightarrow -\infty$ as $r \rightarrow -\infty$. 
For a given set $\mathcal{C}  \subset \mathbb{R}^n$, $\partial \mathcal{C}$ and Int$(\mathcal{C})$ denote its boundary and interior, respectively.

We consider a nominal nonlinear control affine and disturbed nonlinear control affine systems given by 
\begin{align}
\label{system}
    \dot{x} = f(x) + g(x) u, \\
\label{sysdist}
    \dot{x} = f(x) + g(x) u + g(x) d(t, x),
\end{align}
where $x \in X \subset \mathbb{R}^n$, $u \in U \subset \mathbb{R}^m$ is the admissible control input, $d: \mathbb{R}^+_0 \times X \rightarrow D \subset \mathbb{R}^m$ is the time varying, essentially bounded disturbance within a set $D$, and $f: X \rightarrow \mathbb{R}^n$, $g: X \rightarrow \mathbb{R}^{n \times m}$ are locally Lipschitz. We consider the disturbed systems with a matched disturbance for notation simplicity and compatibility with the ISSf-CBF; however, our method can also be adapted to the unmatched disturbance input with a slight modification (see Segway platform example).
\begin{assumption}
\label{asnew}
There exists constants $L_t,L_x \in \mathbb{R}^+$ such that $ \| d(t, x)- d( \tau, y) \| \leq L_t \| t - \tau \| + L_x \| x -y \|~\forall x, y \in X;~ t, \tau \in \mathbb{R}^+$.
\end{assumption} 
Assumption~\ref{asnew} implies that $d(t,x)$ is locally Lipschitz continuous. This is a common assumption in robust control literature~\cite{zhao2020adaptive}.

\subsection{Stability and Control Lyapunov Functions}
CLFs allow the formulation of optimization-based stabilizing controllers, and exponential stability requirements can be reduced to finding a CLF for system (\ref{system}). 
Therefore, CLFs are useful to represent closed-loop control objectives in a CLF-CBF-QP, for instance, reaching a target set \cite{ames2019control, artstein1983stabilization}. 
\begin{definition}[Exponential stability]
The equilibrium point, $x = 0$, of the nonlinear system (\ref{system}) is \textit{exponentially stable} if there are constants $\beta_1, \beta_2, \beta_3 \in \mathbb{R}^+$ such that ${\|x(0)\| \leq \beta_1} \implies \|x(t)\| \leq \beta_2 e^{-\beta_3 t} \|x(0)\|~~ \forall t \geq 0$.
\end{definition}
\begin{definition}[Control Lyapunov function \cite{kolathaya2018CLF, molnar2021model}]
For the disturbed system (\ref{sysdist}), a continuously differentiable function $V : X \rightarrow \mathbb{R}^+_0$ is an \textit{exponentially stabilizing control Lyapunov function}, if there exists constants $\zeta_1, \zeta_2, \lambda \in \mathbb{R}^+$ such that $\forall x \in X:$ $\zeta_1 \|x\|^2 \leq V(x) \leq \zeta_2 \|x\|^2$,
\begin{align}
\nonumber
     \inf_{u \in U} \sup_{d \in D} \Big ( {\dot{V}(x, u, d)}  \triangleq 
      \underbrace{ \dfrac{\partial V}{\partial x} f(x) }_{L_fV(x)} + \underbrace { \dfrac{\partial V}{\partial x}g(x) }_{L_gV(x)} (u + d) & \Big ) \\
      \label{CLF}
      \leq -\lambda V(x),
\end{align}
where $L_f V(x) : X \rightarrow \mathbb{R}$, $L_g V(x) : X \rightarrow \mathbb{R}^m$ are the Lie derivatives of $V(x)$ with respect to $f(x)$, $g(x)$, respectively.
\end{definition}

In real-time applications, external disturbances such as external load and friction may deteriorate the stability or safety of dynamical systems.
In such cases, the definition of a CLF in~(\ref{CLF}) can be extended to input-to-state stabilizing CLF (ISS-CLF) according to the disturbance input $d(t)$ in~(\ref{sysdist}).
The ISS-CLF formula is generally defined for state-independent matched disturbance \cite{kolathaya2018CLF}. Therefore we consider the state-independent essentially bounded disturbance $d(t)$ for the ISS-CLF definition. 
\begin{definition}[Input-to-state stabilizing CLF \cite{kolathaya2018CLF}] For the disturbed system (\ref{sysdist}), a continuously differentiable function $V : X \rightarrow \mathbb{R}^+_0$ is an \textit{exponential input-to-state stabilizing control Lyapunov function (ISS-CLF)} with respect to the essentially bounded disturbance $d(t)$, {$\sup_{t} {\|d(t)\|} < \infty$}, if there exists $\zeta_1, \zeta_2, \lambda \in \mathbb{R}^+,~\iota \in \mathcal{K}_\infty$, such that $\forall x \in X:$ $\zeta_1 \|x\|^2 \leq V(x) \leq \zeta_2 \|x\|^2$,
\begin{equation}
\label{ISSCLF}
    \inf_{u \in U} { L_f V(x)  +  L_g V(x) (u + d)  \leq -\lambda V(x) + \iota (\|d\|_\infty)}.
\end{equation}
\end{definition}

Given a $V(x)$ and $\lambda \in \mathbb{R}^+$ for (\ref{sysdist}), we define the set of exponentially stabilizing controllers as
\begin{equation}
\label{CLFkx}
 K_\text{CLF}(x, d) \triangleq \left\{ u \in U  \big | \dot{V}(x, u, d) \leq -\lambda V(x)~ \forall d \in D  \right\},
\end{equation}
which states that robust exponential stability can be achieved by synthesizing a control input that applies the CLF condition~(\ref{CLF}) to the disturbed system~(\ref{sysdist}).

\subsection{Safety and Control Barrier Functions}
Control barrier functions are a useful tool for rendering the set $\mathcal{C}  \subset \mathbb{R}^n$ as forward invariant throughout its state-space.  
We note that set $\mathcal{C}$ is forward invariant if, for every initial condition $x(0) \in \mathcal{C}$, the solution of (\ref{system}) satisfies $x(t) \in \mathcal{C}$ $\forall t \geq 0$. We consider a set $\mathcal{C} \subset X \subset \mathbb{R}^n$ defined as a 0-superlevel set of a continuously differentiable function $h(x): X \rightarrow \mathbb{R}$ such that
\begin{align}
\label{CBF1}
    \mathcal{C} \triangleq \left\{ x \in X \subset \mathbb{R}^n : h(x) \geq 0 \right\}, \\
    \label{CBF12}
    \partial \mathcal{C} \triangleq \left\{ {x \in X \subset \mathbb{R}^n} : h(x) = 0 \right\}, \\
    \label{CBF13}
    \text{Int}(\mathcal{C}) \triangleq \left\{ {x \in X \subset \mathbb{R}^n} : h(x) > 0 \right\}.
\end{align}
The nominal closed-loop system (\ref{system}) is safe on the set $\mathcal{C}$ if $\mathcal{C}$ is forward invariant \cite{ames2019control}. 
CBFs can be utilized to design safe controllers for system (\ref{system}) with respect to set $\mathcal{C}$. 
\begin{definition}[Control barrier function \cite{ames2019control}]
Let $\mathcal{C} \subset X $ be the 0-superlevel set of a continuously differentiable function $h : X \rightarrow \mathbb{R}$  defined by (\ref{CBF1})-(\ref{CBF13}). 
Then, $h(x)$ is a CBF for system (\ref{sysdist}) on $\mathcal{C}$ if there exists $\alpha \in \mathcal{K}_{\infty, e}$ such that $\forall x \in \mathcal{C}$:
\begin{align}
\nonumber
   \sup_{u \in U} \inf_{d \in D} \Big ( {\dot{h}(x, u, d)}  \triangleq 
   \underbrace{ \dfrac{\partial h}{\partial x} f(x) }_{L_f h(x)} + \underbrace { \dfrac{\partial h}{\partial x}g(x) }_{L_g h(x)} (u+d) &  \Big ) \\
\label{CBF}
   \geq -\alpha (h(x)).
\end{align}
\end{definition}

Given an $h(x)$, $\alpha \in \mathcal{K}_{\infty, e}$ for system (\ref{sysdist}), we define the set of robust safe controllers as
\begin{equation}
\label{CBFkx} 
    K_\text{CBF}(x, d) \triangleq \left\{ u \in U  \big |  \dot{h}(x, u, d) \geq - \alpha (h(x)) ~ \forall d \in D \right\}.
\end{equation}

\begin{definition}[Exponential CBF (ECBF) \cite{nguyen2016exponential}]
\label{D1}
Let $\mathcal{C} \subset X$ be the 0-superlevel set of an $r$-times continuously differentiable function $h : X \rightarrow \mathbb{R}$ such that $L_g L^{r-1}_f h(x) \neq 0$ and $L_g L_f h(x)$ $=$ $L_g L^{2}_f h(x)$ $= \dots =$ $L_g L^{r-2}_f h(x) $ $=$ $0$ $\forall x \in \mathcal{C}$.
Then, $h(x)$ is an ECBF, which is special form of the higher order CBF \cite{xiao2019control}, for system (\ref{sysdist}) on $\mathcal{C}$ if there exists a row vector $K_{\alpha} \in \mathbb{R}^r$ such that $\forall x \in \mathcal{C}$:
\begin{align}
\nonumber
   \sup_{u \in U} \inf_{d \in D} \Big (  {h^r(x, u, d)} \triangleq 
   L^{r}_f h(x) + L_g L^{r-1}_f h(x) & (u + d)  \Big ) \\
   \label{ECBF}
   \geq -K_{\alpha} \eta_b (x),
\end{align}
where $\eta_b (x) = \big [ h(x) \ \dot{h}(x) \ \ddot{h}(x) \ \cdots \  h^{r-1} (x) \big ]^T $. 
\end{definition}
Note that $K_{\alpha}$ in the ECBF definition should satisfy certain specific properties; therefore, we refer the interested readers to \cite{nguyen2016exponential} for details.

\begin{definition}[Input-to-state safe CBF \cite{kolathaya2018input} ]
Let $\mathcal{C} \subset X$ be the 0-superlevel set of a continuously differentiable function $h : X \rightarrow \mathbb{R}$. 
Then, $h(x)$ is an ISSf-CBF for disturbed system (\ref{sysdist}) on $\mathcal{C}$ if there exists $\alpha \in \mathcal{K}_{\infty, e}$, $\iota \in \mathcal{K}_{\infty}$ such that $\forall x \in \mathcal{C}$:
\begin{equation}
\label{ISSfCBF}
     \sup_{u \in U} L_f h(x)  +  L_g h(x) u  \geq - \alpha (h(x)) - \iota (\|d\|_\infty).
\end{equation}
\end{definition}

When valid CLFs and CBFs are given for system~(\ref{system}), the safety constraints can be enforced with a relaxed CLF constraint to compute pointwise safe control inputs via following CLF-CBF-QP:   
\begin{argmini*}|s|
{u \in U, \delta \in \mathbb{R}}{\|u-k(x)\|^2 + p \delta^2}
{\label{CLF-CBF-QP}}
{u^*(x)=}
\addConstraint{\dot{h}(x, u)  \geq - \alpha (h(x)) }
\addConstraint{\dot{V}(x, u)  \leq -\lambda V(x) + \delta}
\end{argmini*}
where $k(x)$ is a locally Lipschitz continuous baseline control law, $\delta \in \mathbb{R}$ is a relaxation variable that is penalized by a constant $p \in \mathbb{R}^+$. 
In \cite{xu2015robustness}, this QP-based controller
is shown to generate Lipschitz continuous controllers.

Similarly, one can combine ISS-CLF~(\ref{ISSCLF}) and ISSf-CBF~(\ref{ISSfCBF}) constraints to synthesize pointwise safe control inputs for the disturbed system~(\ref{sysdist}) via the following~ISSf-CBF-QP~\cite{kolathaya2018input, alan2021safe}:
\begin{argmini*}|s|[3]<b>
{u \in U,~\delta \in \mathbb{R}}{\|u-k(x)\|^2 + p \delta^2}
{\label{ISSf-CBF-QP}}
{u^*(x)=}
\addConstraint{\dot{h}(x, u) }{\geq - \alpha (h(x)) + \epsilon \|L_g h(x)\|^2}
\addConstraint{\dot{V}(x, u)}{\leq -\lambda V(x) + \delta}
\end{argmini*}
where $\epsilon \in \mathbb{R}^+$ is a user-defined constant.
Safe controller synthesizing using ISSf-CBF-QP is a way to handle unmodeled system dynamics. 
However, this QP may conservatively ensure the safety requirements for system~(\ref{sysdist}) depending on the selected $\epsilon$ parameter. 

\subsection{High-Gain Input Disturbance Observer}
In order to define robust, exponentially stabilizing, safe controllers using (\ref{CLFkx}) and (\ref{CBFkx}) we need to measure the time-varying disturbance $d(t, x) \in D$ that is not directly available in real-time applications. 
Using a disturbance observer framework, our objective is to replace $d(t, x)$ with an estimated disturbance term $\hat{d}$ and the upper bound of the associated estimation error. 
Note that the dependence on time $t$ will be omitted for simplicity throughout the rest of the paper, and it will be used only if necessary. 

Specifically, we consider a high-gain input disturbance observer that is proposed in \cite{stotsky2002application}.
Let us define a first-order dynamical system
\begin{equation}
\label{DOB1}
    \dot{z}_d= v_d + w_d,
\end{equation}
where $z_d \in \mathbb{R}$ and $v_d \in \mathbb{R}$ are known or measured variables, and $w_d \in \mathbb{R}$ is the unknown time-varying unmodelled dynamics or disturbance input of the system that needs to be estimated.
Define estimated disturbance $\hat{w}_d \in \mathbb{R}$ as
\begin{equation}
\label{DOBe}
    \hat{w}_d = k_d z_d - \varepsilon_d,  
\end{equation}
where $k_d \in \mathbb{R}^+$ is the disturbance observer gain to be tuned, and
$\varepsilon_d \in \mathbb{R}$ is an auxiliary variable satisfying \begin{equation}
\label{DOBeps}
    \dot{\varepsilon}_d = - k_d \varepsilon_d + k_d v_d + k_d^2 z_d.
\end{equation}
Then, the error dynamics of this disturbance estimation method, $e_d = (w_d - \hat{w}_d) \in \mathbb{R}$, is obtained as
\begin{equation}
\label{DOBer}
    e_d  =  w_d + \varepsilon_d - k_d z_d.
\end{equation}
\begin{definition}[Estimation error quantified observer \cite{wang2021observer}]
A disturbance observer is called an \textit{estimation error quantified observer} for system~(\ref{DOB1}) if it generates a disturbance estimation~$\hat{w}_d$ with an error bound~$\|e_d\|$ such that~$\forall t \geq 0$,
\begin{equation}
\label{EEQ}
   \|e_d\| \leq M_d(t, w_d, \hat{w}_d), 
\end{equation}
where $M_d(t, w_d, \hat{w}_d): \mathbb{R}_0^+ \times \mathbb{R} \times \mathbb{R} \rightarrow \mathbb{R}_0^+$. 
\end{definition}
If we assume that $\dot{w}_d$ is bounded as ${\|\dot{w}_d\|} \leq b_1$,
then the high-gain disturbance observer model given in (\ref{DOB1})-(\ref{DOBer}) is an estimation error quantified disturbance observer with the following error bound \cite{stotsky2002application}:
\begin{equation}
\label{DOBE}
   \|e_d\| \leq \sqrt{e_d(0)^2 e^{-k_d t} + {b_1^2}/{k_d^2}}.
\end{equation}

\section{Disturbance~Observer-Based Safety-Critical Control}
In this section, we address the issue of having disturbances by proposing a new disturbance observer-based robust safety-critical framework.
Specifically, we adapt the high-gain disturbance observer scheme to estimate the time-varying effect of disturbance on the time derivative of CBF with the associated error bound. 
Next, we use the estimated part of the CBF constraint and error bound to construct a robust safety constraint. 

The CLF-CBF-QP for disturbed system (\ref{sysdist}) can be formulated using the linear constraints of the exponentially stabilizing and safe controller sets in (\ref{CLFkx})~and~(\ref{CBFkx}) as
\begin{argmini}|s|[3]<b>
{u \in K_\text{CBF},~\delta \in \mathbb{R}}{\|u-k(x)\|^2 + p \delta^2}
{\labelOP{rCBF0}}
{u^*(x)=}
{\labelOP{rCBF}}
\addConstraint{L_f h(x) + L_g h(x)u + L_g h(x)d \geq - \alpha (h(x))  }
\addConstraint{L_f V(x) + L_g V(x)u + L_g V(x)d  \leq -\lambda V(x) + \delta}
\end{argmini}
where $L_g h(x)d$, $ L_g V(x)d$ are unknowns since they depend on immeasurable disturbance $d$.
The objective function is set to modify desired feedback control input minimally.

If we consider $h : X \rightarrow \mathbb{R}$ to be a CBF for system (\ref{sysdist}) on set $\mathcal{C}$, the time derivative of $h(x)$ is given by
\begin{equation}
\label{doth}
  \dot{h}(x, u, d) = \underbrace{ L_f h(x) + L_g h(x)u}_{a(x, u)} + \underbrace{ L_g h(x)d }_{b(x, d)},
\end{equation}
where $a(x, u) : X \times U  \rightarrow \mathbb{R}$ is the known part of $\dot{h}(x, u, d)$, and $b(x, d) : X \times D  \rightarrow \mathbb{R}$ needs to be estimated.
The effect of the disturbance on the $\dot{h}(x, u, d)$ is obvious.
Furthermore, the first-order system dynamics of~(\ref{doth}) is in the form of an observation problem in standard format~(\ref{DOB1}).
Therefore, the high-gain input disturbance observer can be defined for estimation of $b(x, d)$, i.e., $\hat{b}(x, u, d)$, as
\begin{align}
\label{DOBb}
  &  \hat{b}(x, u, d)  = k_b {h}(x) - \varepsilon_b,  \\
\label{DOBe1}
  &  \dot{\varepsilon}_b  = - k_b \varepsilon_b + k_b a(x, u) + k_b^2 {h}(x), \\
    \label{DOBb1}
  &  e_b  = b(x, d) -  \hat{b}(x, u, d) =  b(x, d) + \varepsilon_b - k_b {h}(x),
\end{align}
where $k_b \in \mathbb{R}^+$ is the disturbance observer gain, $e_b \in \mathbb{R}$ is the estimation error dynamic, and $\varepsilon_b \in \mathbb{R}$ is the auxiliary variable.
\begin{assumption}
\label{as1}
There exists a constant $b_h \in \mathbb{R}^+_0$ such that ${\|\dot{b}(t, x, u, d)\|} \leq b_h$, where $\dot{b}(t, x, u, d)$ is given by
\begin{equation}
\label{bh}
  \dot{b}(t, x, u, d) = \dfrac{\partial b}{\partial t} + \dfrac{\partial b}{\partial x}
  (f(x) + g(x) u + g(x) d(t, x)).
\end{equation} 
\end{assumption}
\vspace{1ex}
Note that $L_g h(x)$ and ${\partial d(t, x)}/{\partial t}$ are bounded on the set $X$; therefore, ${\partial b}/{\partial t} = L_g h(x) {\partial d(t, x)}/{\partial t}$ is bounded.
Since $h(x)$ is continuously differentiable, and $d(t, x)$, $L_g h(x) $ are locally Lipschitz, $b(t, x, d) = L_g h(x) d(t, x)$ is locally Lipschitz because the sum or product of two Lipschitz functions is also Lipschitz on a bounded set \cite{xu2015robustness}.
Similarly, by locally Lipschitz properties of $f(x),~g(x),~d(t, x)$, we can show that $\dot x = f(x) + g(x) u + g(x) d(t, x)$ is locally Lipschitz on $X$ with a locally Lipschitz $u$. 
Therefore, Assumption~\ref{as1} can be satisfied by utilizing the fact that a Lipschitz function is bounded on a bounded domain.
\begin{assumption} 
\label{as2}
The CBF and CLF have relative degree one, i.e., $L_g h(x) \neq 0$, $L_g V(x) \neq 0$ $\forall x \in X$ in (\ref{rCBF}).
\end{assumption}
Note that, firstly, we consider the relative degree one systems using Assumption \ref{as2}, then we extend our results to the higher relative degree systems via an ECBF.

\begin{theorem}
\label{t1}
The estimation error dynamics of the disturbance observer given in (\ref{DOBb1}), under Assumption \ref{as1} and Assumption \ref{as2}, convergences to a set defined by 
\begin{equation}
\label{erb}
  \| e_b(t, x, u, d) \| \leq \sqrt{ ( e_b(0)^2 - {b_h^2}/{k_b^2}) e^{-k_b t} + {b_h^2}/{k_b^2}}.
\end{equation}
\end{theorem}
\begin{proof}
In order to prove the convergence of the proposed disturbance observer to its steady-state value and derive an estimation error bound, we choose a Lyapunov function $V_d : \mathbb{R} \rightarrow \mathbb{R}^+_0$ as follows by omitting the dependence on ($t, x, u, d, \hat{d} $) for simplicity:
\begin{equation}
\label{Lya}
V_d = \dfrac{1}{2}(e_b)^2 = \dfrac{1}{2}(b - \hat{b})^2.
\end{equation}
Then, time derivative of $V_d$ is given by 
\begin{equation}
\label{Lyad1}
\dot{V}_d =  (b - \hat{b}) (\dot{b} - \dot{\hat{b}}) = \dfrac{1}{2} \dfrac{d(b- \hat{b})^2}{dt}.
\end{equation}
Substituting (\ref{DOBb}) and its time derivative into (\ref{Lyad1}) yields
\begin{equation}
\label{Lyad2}
\dot{V}_d =  (b + \epsilon_b - k_b h) (\dot{b} - k_b \dot{h} + \dot{\epsilon}_b ),
\end{equation}
and, substituting (\ref{doth}) and (\ref{DOBe1}) into (\ref{Lyad2}), we have 
\begin{align}
\label{Lyad}
 \dot{V}_d & =  (b + \epsilon_b - k_b h) (\dot{b} - k_b b - k_b \epsilon_b + k_b^2 h ) \\
& =  (b - \hat{b}) (\dot{b} - k_b (b- \hat{b})) = -k_b (b - \hat{b})^2 + (b - \hat{b}) \dot{b}. \nonumber
\end{align}
Noting that,
\begin{equation}
\label{Lyade1}
-k_b (b - \hat{b})^2 + (b - \hat{b}) \dot{b} \leq -k_b (b - \hat{b})^2 + \| b - \hat{b} \|b_h,
\end{equation}
we obtain
\begin{equation}
\label{Lyade2}
\dfrac{1}{2} \dfrac{d(b- \hat{b})^2}{dt} \leq -k_b (b - \hat{b})^2 + \| b - \hat{b} \|b_h.
\end{equation}
Therefore, we need to define an upper bound for $\| b - \hat{b} \|b_h$.
Now, consider the following inequality
\begin{equation}
\label{Lyade3}
 (k_b \| b - \hat{b} \| - b_h)^2 = k_b^2 \| b - \hat{b} \|^2 -2 \| b - \hat{b} \| k_b b_h   + b_h^2 \geq 0,
\end{equation}
which results
\begin{equation}
\label{Lyade4}
\| b - \hat{b} \|b_h \leq \dfrac{k_b \| b - \hat{b} \|^2}{2} + \dfrac{b_h^2}{2 k_b}.
\end{equation}
Substituting this upper bound into (\ref{Lyade2}), we obtain
\begin{equation}
\label{Lyade5e}
2 \dot{V}_d = \dfrac{d(b- \hat{b})^2}{dt} \leq -k_b (b - \hat{b})^2 + \dfrac{b_h^2}{k_b}.
\end{equation}
Integration of (\ref{Lyade5e}) yields the following inequality
\begin{equation}
\label{Lyade7}
\| b - \hat{b} \| \leq  \sqrt{ e^{b_\alpha}  e^{-k_b t} + {b_h^2}/{k_b^2}},
\end{equation}
where ${b_\alpha} \in \mathbb{R}$. Finally, solving (\ref{Lyade7}) for $e^{b_\alpha}$ with initial conditions leads to
\begin{equation}
\label{Lyade8}
 e^{b_\alpha}  \geq (b(0) - \hat{b}(0))^2 -b_h^2/k_b^2 
\end{equation}
for $t > 0$; therefore, we can replace $e^{b_\alpha}$ with a constant using~(\ref{Lyade8}) as $e^{b_\alpha}= (b(0) - \hat{b}(0))^2 -b_h^2/k_b^2$ that leads to
\begin{equation}
\label{Lya3}
  \| b - \hat{b} \| \leq \underbrace{ \sqrt{ \big ( (b(0) - \hat{b}(0))^2 -b_h^2/k_b^2 \big ) e^{-k_b t} + {b_h^2}/{k_b^2}}}_{M_b(t, x, u, d, \hat{d})},
\end{equation}
which is the statement of the theorem.
\end{proof}
\begin{remark}
\label{Re1}
Since $( b(0) - \hat{b}(0) )^2 > (b(0) - \hat{b}(0))^2 -b_h^2/k_b^2$, Theorem~\ref{t1} provides a less conservative upper bound with a convergence guarantee for the estimation error dynamics than (\ref{DOBE}), derived in \cite{stotsky2002application}.
\end{remark}

The proposed high-gain disturbance observer is an estimation error quantified disturbance observer for the time derivative of $h(x)$.
Moreover, (\ref{Lya3}) implies that $\hat{b}(x, u, d) \rightarrow b(x, d)$ if $k_b \rightarrow \infty$, and if $(L_g h)^{-1}$ exists, one can easily compute the estimated disturbance as
\begin{equation}
\label{do}
  \hat{d} = (L_g h)^{-1} \hat{b}(x, d).
\end{equation}
Hence, if there exists a constant $b_d \in \mathbb{R}^+_0$ such that ${\|\dot{d}\|} \leq b_d$, we can derive a bound for $d_e$,
\begin{equation}
\label{dde}
    \hat{d} = d - d_e,
\end{equation}
as $\| d_e \| \leq M_d(t, x, u, d, \hat{d}) := \| (L_g h)^{-1} \|M_b(t, x, u, d, \hat{d}) $.  
Then, plugging in for $d$ in the constraints of optimization problem (\ref{rCBF}) yields
\begin{align}
\label{deCBF}
  L_f h(x) + L_g h(x) (u + \hat{d} + d_e ) \geq - \alpha (h(x)),  \\
    L_f V(x) + L_g V(x) (u + \hat{d} + d_e ) \leq -\lambda V(x) . 
\end{align}

\begin{lemma}
\label{l1}
Consider the disturbed system (\ref{sysdist}) with an estimated error quantified disturbance observer that provides $\hat{d}$ with an error bound $\| d_e \| \leq M_d(t, x, u, d, \hat{d})$.
Suppose that a safe set $\mathcal{C} \subset X$ and a 0-superlevel set of the continuously differentiable function $h(x): X \rightarrow \mathbb{R}$, and $\alpha \in \mathcal{K}_{\infty, e}$ are given for the nominal system (\ref{system}).
If a control signal $u \in U$ satisfies
\begin{equation}
\label{the1}
      L_f h(x)  +  L_g h(x) u 
     + \underbrace{ L_g h(x) \hat{d}}_{\hat{b}} -  \underbrace{ \|L_g h(x)\|M_d}_{M_b}  \geq - \alpha (h(x)) , 
\end{equation}
then the robust CBF constraint in (\ref{rCBF}) is also guaranteed.
\end{lemma}
\begin{proof}
Our objective is to show that $L_g h(x)d$ in (\ref{rCBF}) is an upper bound for $L_g h(x)  \hat{d} -  \|L_g h(x)\|M_d $ $\forall t \geq 0$. 
We have
\begin{align}
\label{pr1}
    L_g h(x) d & =  L_g h(x) (\hat{d} + d_e) =  L_g h(x) (\hat{d} + d - \hat{d}) \nonumber \\
   & \geq  L_g h(x) (\hat{d} )  - \| L_g h(x) \| \Big (\| d - \hat{d}  \| \Big ) \nonumber \\
   & \geq \underbrace{ L_g h(x) (\hat{d} )}_{\hat{b}}  - \underbrace{ \| L_g h(x) \| M_d(t, d, \hat{d}) }_{M_b},
\end{align}
which means that (\ref{the1}) $\implies$ (\ref{CBFkx}).
\end{proof}
\begin{remark}
\label{R1}
Lemma~\ref{l1} provides a sufficient robust safety condition via a modified CBF, and if $\hat{b} \rightarrow b$, (\ref{the1}) becomes equivalent to the robust safety constraint given in (\ref{doth}) with a sufficiently large $k_b$. However, this high gain amplifies the measurement noise that affects ${h}(x, u, d)$, $a(x, u)$ variables in (\ref{DOBb}), (\ref{DOBe1}), (\ref{DOBb1}). Therefore, we need to choose an appropriate $k_b$ parameter by considering steady-state estimation error, i.e., $\sqrt{{b_h^2}/{k_b^2}}$, and the effects of the sensor noises.
\end{remark}

\begin{remark}
The proposed method is also applicable to high relative degree CBF constraints.
To show this, let consider a higher relative degree disturbed system (\ref{sysdist}) with an ECBF $h(x)$ defined in Definition \ref{D1}, where $L_g h(x) = 0$ $~\forall x \in X$. 
The $r^{th}$-order time derivative of $h(x)$ is given by
\begin{equation}
\label{ECBF2}
 \underbrace{ h^r(x, u, d) }_{\big (\dot{h}^{r-1}(x, u, d) \big)} = \underbrace{ L^{r}_f h(x) + L_g L^{r-1}_f h(x) u}_{a_e(x, u)}  + \underbrace{ L_g L^{r-1}_f h(x) d}_{b_e(x, d)},
\end{equation}
which is in the form of the first-order dynamical system (\ref{DOB1}); therefore, we can adapt the proposed input disturbance observer scheme to estimate the unknown dynamics $b_e(x, d)$.
Again, if there exists a constant $b_E \in \mathbb{R}^+_0$ such that ${\|\dot{b}_e(x, d)\|} \leq b_E $, then a high-gain disturbance observer can be proposed to estimate $b_e(x, d)$ with an error bound $\| b_e - \hat{b}_e\| \leq M_{b_e}(t, x, u, d, \hat{d})$. 
Finally, the robust ECBF constraint is defined as
\begin{equation}
\label{ECBF3}
  L^{r}_f h(x) + L_g L^{r-1}_f h(x) u + \underbrace{ L_g L^{r-1}_f h(x) \hat{d}}_{\hat{b}_e(x, d)} - M_{b_e} \geq -K_{\alpha} \eta_b(x),
\end{equation}
where $K_{\alpha},~\eta_b(x)$ are given in (\ref{ECBF}).
\end{remark}
 
Note that proposed disturbance observer framework estimates the effect of the disturbance on the time derivative of $h(x)$. Therefore, Lemma~\ref{l1} provides a sufficient robust CBF condition even if $(L_g h)^{-1}$ is not exactly known.
However, without the estimated disturbance $\hat{d}$ and the associated error bound, the robust CLF constraint cannot be defined.
To address this issue, consider the time derivative of $V(x) : \mathbb{R}^n \rightarrow \mathbb{R}$ given by
\begin{equation}
\label{dotV}
  \dot{V}(x, u, d) = \underbrace{ L_f V(x) + L_g V(x)u}_{a_V(x, u)} + \underbrace{ L_g V(x)d }_{b_V(x, d)},
\end{equation}
where $b_V(x, d)$ needs to be estimated. 
If there exists a constant $b_L \in \mathbb{R}^+_0$ such that ${\|\dot{b}_V(x, d)\|} \leq b_L$, then we can design a disturbance observer to estimate $b_V(x, d)$ with an error bound $\| b_V - \hat{b}_V\|$ $\leq M_{b_V}(t, x, u, d, \hat{d})$, $\| d_e\| =\| d - \hat{d}\|$  $\leq M_{d_V}(t, x, u, d, \hat{d}) = \| (L_g V)^{-1} \|M_{b_V}(t, x, u, d, \hat{d}) $, 
which leads to following lemma to provide a sufficient robust CLF condition for system (\ref{sysdist}). 

\begin{lemma}
\label{l2}
Consider the disturbed nonlinear system (\ref{sysdist}) with an estimated error quantified observer that provides $\hat{d}$ with an estimation error bound $\| d_e \|$~$\leq M_{d_V}(t, x, u, d, \hat{d})$.
Suppose that a continuously differentiable  exponentially stabilizing CLF function $V : X \rightarrow \mathbb{R}^+_0$, and $\lambda \in \mathbb{R}^+$ is given for nominal system (\ref{system}).
If a control signal $u \in U$ satisfies
\begin{equation}
\label{the2}
      L_f V(x)  +  L_g V(x) u + \underbrace{ L_g V(x) \hat{d}}_{\hat{b}_V} + \underbrace{ \|L_g V(x)\|M_{d_V}}_{M_{b_V}} \leq -\lambda V(x),
\end{equation}
then the robust CLF constraint in (\ref{rCBF}) is also guaranteed.
\end{lemma}
\begin{proof}
Our objective is to show that $L_g V(x)d$ in (\ref{rCBF}) is a lower bound of $L_g V(x)  \hat{d} +  \|L_g V(x)\|M_d $ $~\forall t \geq 0$.
We also have
\begin{align}
\label{pr2}
    L_g V(x) d & =  L_g V(x) (\hat{d} + d_e) =  L_g V(x) (\hat{d} + d - \hat{d}) \nonumber \\
   & \leq  L_g V(x) (\hat{d} )  + \| L_g V(x) \| \Big (\| d - \hat{d}  \| \Big ) \nonumber \\
   & \leq \underbrace{ L_g V(x) (\hat{d} )}_{\hat{b}_V}  + \underbrace{\| L_g V(x) \| M_d(t, d, \hat{d})}_{M_{b_V}} ,
\end{align}
which means that (\ref{the2}) $\implies$ (\ref{CLFkx}).
\end{proof}

By Lemma \ref{l1} and Lemma \ref{l2}, the pointwise safe controller with a valid CBF for disturbed nonlinear system (\ref{sysdist}) is obtained by the following robust CLF-CBF-QP:
\begin{argmini}|s|[3]<b>
{u \in K_\text{CBF},~\delta \in \mathbb{R}}{\|u-k(x)\|^2 + p \delta^2}
{\labelOP{robCBF0}}
{u^*(x)=}
{\labelOP{robCBF}}
\addConstraint{L_f h(x)  +  L_g h(x) u + \hat{b} - M_b  \geq - \alpha (h(x)) }
\addConstraint{L_f V(x)  +  L_g V(x) u + \hat{b}_V + M_{b_V} \leq -\lambda V(x) + \delta}
\end{argmini}
Note that, for higher relative degree systems, the CBF constraint in (\ref{robCBF}) needs to be replaced with the robust ECBF constraint given in (\ref{ECBF3}). 

Finally, if $(L_g h)^{-1}$ exists, we do not need to use another disturbance observer for the robust CLF constraint since $\hat{d}$ can be obtained via (\ref{do}).
Then, the constraints in (\ref{robCBF}) are replaced as in the following robust CLF-CBF-QP:
\begin{align}
{\label{deCBFm}}
\begin{array}{lll}
{u^*(x)= \ }
\displaystyle  \argmin_{u \in K_\text{CBF},~\delta \in \mathbb{R}} \ \ \ {\|u-k(x)\|^2 + p \delta^2}  \\ [2mm]
\textrm{s.t.}  \\
 L_f h(x)  +  L_g h(x) (u + \hat{d}) -  \|L_g h(x)\|M_d  \geq - \alpha (h(x)) \\ [1mm]
L_f V(x)  +  L_g V(x) (u + \hat{d}) + \|L_g V(x)\|M_d   \leq -\lambda V(x) + \delta \
\end{array}
\end{align}
One of the most common ways to compensate for the effects of the input disturbance is to subtract the estimated disturbance from the baseline control signal. Therefore, in order to reject the effects of the disturbance, the objective function of the CLF-CBF-QP in (\ref{robCBF}) can be modified as $\|u-(k(x)-\hat{d})\|^2$.

\section{Simulation Results}

\subsection{Adaptive Cruise Control Example}
In this subsection, we apply the proposed disturbance observer based robust, safe controller design methodology to an adaptive cruise control example \cite{xu2015robustness, zhao2020adaptive}, in which our safety objective is to maintain the safe following distance while cruising at a constant speed.
The system dynamics are in the form (\ref{sysdist}):
\begin{equation}
\label{ACCdyn}
\underbrace{ \begin{bmatrix}
\dot{v_l} \\ 
\dot{v_f} \\ 
\dot{D}
\end{bmatrix} }_{\dot{x}}
=
\underbrace{ \begin{bmatrix}
a_l \\ 
{-F_r/m} \\ 
v_l - v_d
\end{bmatrix} }_{f(x)}
+ 
\underbrace{ \begin{bmatrix}
0 \\ 
1/m \\ 
0
\end{bmatrix} }_{g(x)}
u
+ 
\underbrace{ \begin{bmatrix}
0 \\ 
1/m \\ 
0
\end{bmatrix} }_{g(x)}
d,
\end{equation}
where $v_l \ [m/s]$ and $v_f \ [m/s]$ are the velocity of the lead car and the following car, respectively, and $D \ [m]$ is the distance between the lead and following cars,~ $F_r = f_0 + f_1 v_f + f_2 (v_f)^2 \ [N]$ is the aerodynamic drag, $m \ [kg]$ is the mass of the following car, $a_l$ is the acceleration of the lead car, $d$ is the external disturbance.
The safety constraint requires the following car to keep a safe distance from the lead car as $D \geq v_f \tau_d $, where $\tau_d$ is the desired time headway. 
The CBF  $h(x) =  D - v_f \tau_d$ captures this.
The closed-loop control objective requires cruise at a constant speed that is encoded into the QP via CLF $V(x) = ( v_f - v_d)^2$.
The parameters of the ACC problem simulation are given in Table 1. 
Since the ACC system is in the form of a single-input system and $(L_g h)^{-1} \neq 0$, we use the robust CLF-CBF-QP given in (\ref{deCBFm}). 
The objective of robust CLF-CBF-QP is set to be ${0.5}  u^2 /{m^2} + 0.5 p \delta^2$.
The control signal is constrained as $-0.4 mg \leq u \leq 0.4 mg$.
For the sake of completeness, we compare the proposed disturbance observer based method and ISSf-CBF-QP, where we consider $\epsilon = 2620000$, and CLF-CBF-QP.

For this ACC problem, the effects of the sinusoidal input disturbances $d(t)$ on the CBF constraint (\ref{doth}) and its time derivative (\ref{bh}) are given by $b(t, x) = \tau_d {d}/{m} $ and $\dot{b}(t, x) = {\tau_d \dot{d}}/{m} $, respectively. Substituting the parameters of the system, we obtain $b_h = 2.219$. Then, in order to design the proposed disturbance observer with a sufficiently small estimation error dynamic and short convergence time, we choose $k_b = 100$.
\begin{figure}
\centering
  \includegraphics[width=85mm]{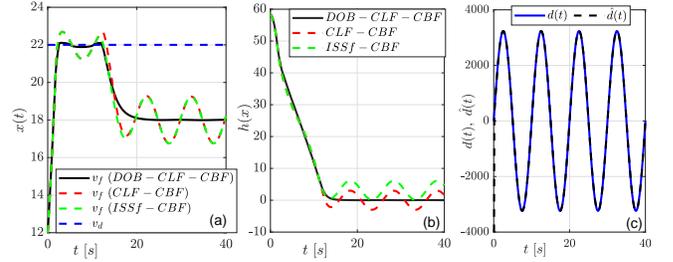}
  \caption{Numerical simulation of the ACC system. (a) Velocity of the lead car and desired velocity. (b) Control barrier function $h(x)$. (c) Estimated and actual disturbances. Simulations are performed with disturbance observer (DOB)-based robust CLF-CBF-QP (proposed method) (black), nominal CLF-CBF-QP (red), and ISSf-CBF-QP (green).} \label{Xall}
\end{figure}
\begin{table}[t]
\begin{center}
\label{table1}
\caption{Parameters in simulation for ACC example.}
\begin{tabular}{ |c|c|c|c|c} 
 \hline
 $m= 1650 \ [kg]$ & $f_0= 0.1 \ [N]$ & $p= \ 100$ \\ 
 $g = 9.81 \ [m/s^2]$ & $f_1= 5 \ [Ns/{m}]$ & $\lambda = 5$ \\ 
 $a_l= 0 \ [m/{s^2}]$ & $f_2= 0.25 \ [Ns^2/{m}]$ & $\alpha(h(x)) = h$ \\ 
 $v_d= 22 \ [m/{s}]$ & $x(0) = [18 \ 12 \ 80]^T$ & $\tau_d= 1.8 \ [s]$ \\ 
 $b_h = 2.22$ & $k_b = 100$ & $d = 0.2 g sin(20 \pi t)$ \\
 \hline
\end{tabular}
\end{center}
\end{table}

Fig.~\ref{Xall} shows the evaluation of $v_f$, $h(x)$, and disturbance estimation accomplishment of the disturbance observer, respectively.
As can be seen from Fig.~\ref{Xall}-(b), the applied disturbance input causes safety violation without the disturbance estimator, i.e., CLF-CBF-QP violates the safety requirement. 
Although the ISSf-CBF provides a safety guarantee, it causes performance deterioration due to the conservative structure of the method, as observed from Fig.~\ref{Xall}-(a) and Fig.~\ref{Xall}-(b).
The safety and control performance of the system is maintained using the proposed disturbance observer based robust safety control algorithm, as shown in Fig.~\ref{Xall}-(a) and Fig.~\ref{Xall}-(b). 
Furthermore, it is observed from  Fig.~\ref{Xall}-(c) that the proposed disturbance observer can effectively estimate actual disturbance. 

\subsection{Segway Platform Example}
The proposed method can also be applied to uncertain systems with high relative-degree safety constraints.
This section demonstrates the proposed algorithm on a Segway platform model on an inclined surface, which causes unmatched disturbance and uncertainties, with a CBF with relative degree two.

We consider a planar Segway model given by Fig.~\ref{Segway}-(a) with the position and pitch angle states $[p \ \theta]^T$ on an unknown surface inclined by $\phi$.
For simplicity, let assume that $\dot{\phi} = 0$, $\ddot{\phi} = 0$.
We have the following dynamics:
\begin{align}
\label{Segwaydyn}
\begin{bmatrix}
m_0 & m L cos(\psi) \\ 
m L cos(\psi) & J_0
\end{bmatrix} 
\begin{bmatrix}
\ddot{p} \\ 
\ddot{\theta} 
\end{bmatrix} 
+ 
\begin{bmatrix}
0 \\ 
- m g L sin(\psi)
\end{bmatrix} \nonumber
\\
+ 
\begin{bmatrix}
b_t/R & -b_t - m L \dot{\theta} sin(\psi)  \\ 
-b_t & b_t R
\end{bmatrix}
\begin{bmatrix}
\dot{p} \\ 
\dot{\theta} 
\end{bmatrix} 
=
\begin{bmatrix}
K_m/R \\ 
-K_m
\end{bmatrix}
u,
\end{align}
where $\psi = \theta + \phi$.
For more details on the descriptions and values of the parameters, see \cite{molnar2021safety}. 
Note that the terms including inclination angle $\phi$ can be viewed as the uncertainty and disturbance.
Choosing the state vector $x = [{p}~\dot{p}~{\theta}~\dot{\theta}]^T$, we can get the dynamics in the form (\ref{sysdist}) with uncertainties as
\begin{equation}
\label{sysunc}
    \dot{x} = f(x) + g(x) u + \underbrace{\Delta f(x, \phi)  + \Delta g(x, \phi)u}_{d(x, u, \phi)},
\end{equation}
where, the effect of inclination uncertainty is lumped into a single unmatched disturbance vector $d(x, u, \phi)$.
The unmodelled dynamics of the system, which is described in \cite{gurriet2018towards}, are also included in $\Delta f(x, \phi)$, $\Delta g(x, \phi)$.
We choose a control barrier function~{$h(x) = \pi/10 - \theta^2$}, which has relative degree two, to encode constraints on the pitch angle.
A Linear Quadratic Regulator (LQR) nominal controller is designed using the linearized model of the system to track the desired path.
To estimate the effect of $d(x, u, \phi)$ on the time derivative of CBF $\dot{h}(x, u, d)$, the proposed disturbance estimation framework can be adapted as
\begin{equation}
\label{dothseg}
  \dot{h}(x, u, d) = \underbrace{ L_f h(x) + L_g h(x)u}_{a(x, u)} + \underbrace{ \dfrac{\partial h}{\partial x}d(x, u, \phi) }_{b_d(x, u, \phi)}.
\end{equation}
We replace $\frac{\partial h}{\partial x}g(x)d$ between equations (\ref{CBF}) and (\ref{deCBFm}) with $\frac{\partial h}{\partial x}d$ in the case of unmatched disturbance. 

The effect of the uncertainty on the time derivative of CBF $b_d(x, u, \phi)$ should be locally Lipschitz to find an upper bound on the time derivative of $b_d(x, u, \phi)$, which is required for the proposed disturbance observer framework.
For the Segway example, by locally Lipschitz properties of $f(x),~g(x),~u$, and continuously differentiability of $h(x)$, we can show that $b_d(x, u, \phi)$ is locally Lipschitz.
Therefore, we can modify the robust CBF constraint in (\ref{robCBF}) as
\begin{equation}
\label{dothmod}
L_f h(x)  +  L_g h(x) u + \hat{b}_d(x, u, \phi) - M_{b_d}  \geq - \alpha (h(x)),
\end{equation}
where $M_{b_d}$ is the estimation error bound of $\hat{b}_d(x, u, \phi)$.
We use only the robust CBF constraint for this example. 
\begin{figure}
  \begin{subfigure}{2.9cm}
    \centering\includegraphics[width=2.9cm]{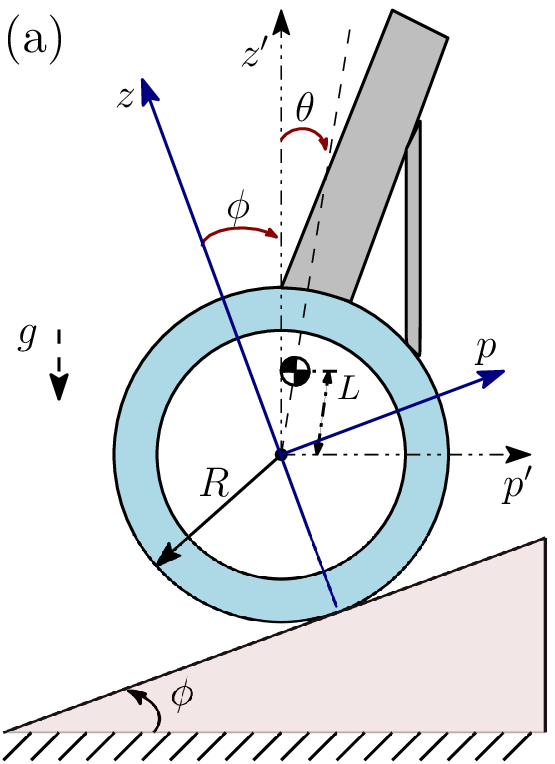}
  \end{subfigure} 
   \begin{subfigure}{5.65cm}
    \centering\includegraphics[width=5.6cm]{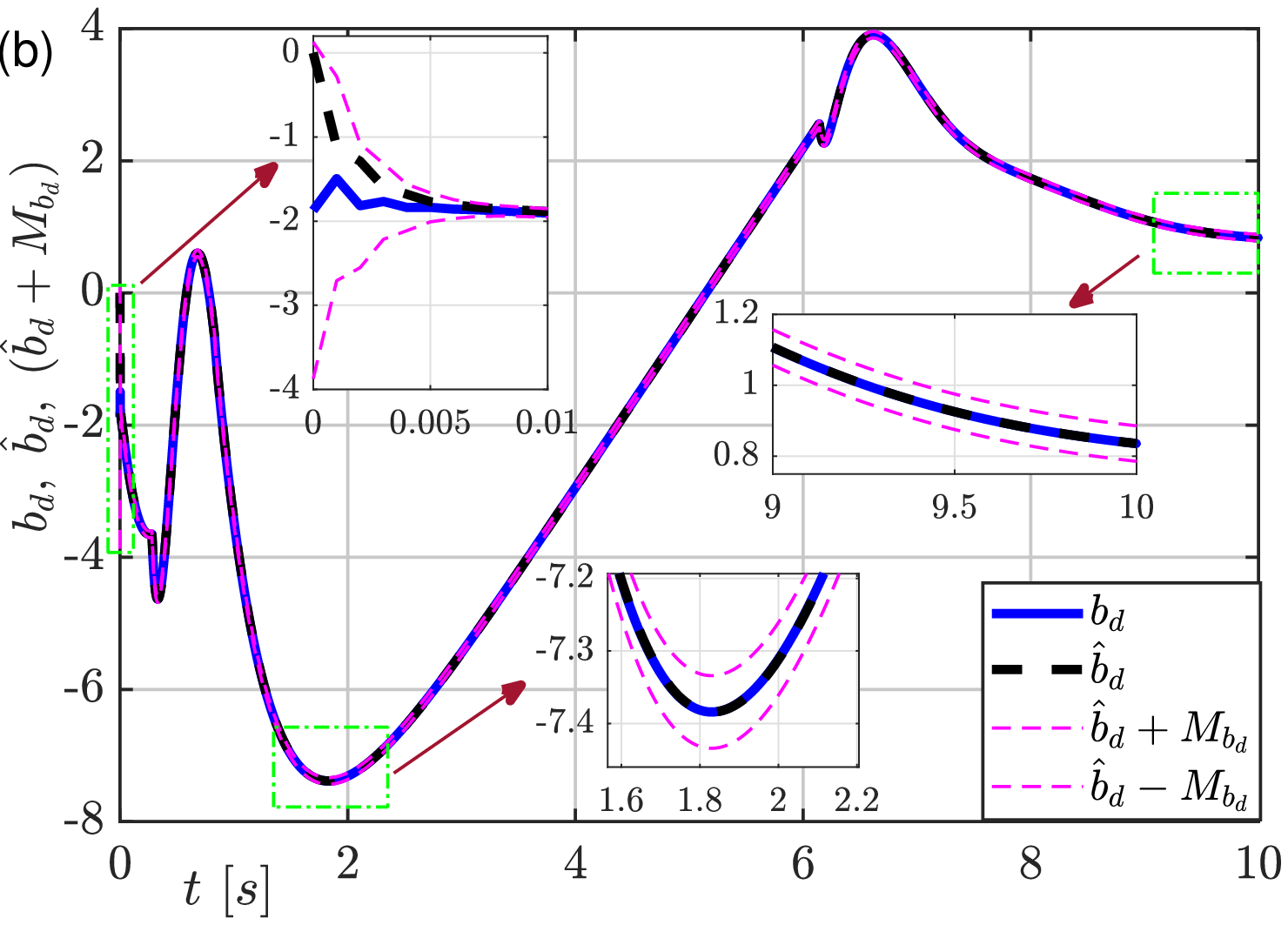}
  \end{subfigure} \\ [0.2cm]
  \begin{subfigure}{8.6cm}
    \centering\includegraphics[width=8.5cm]{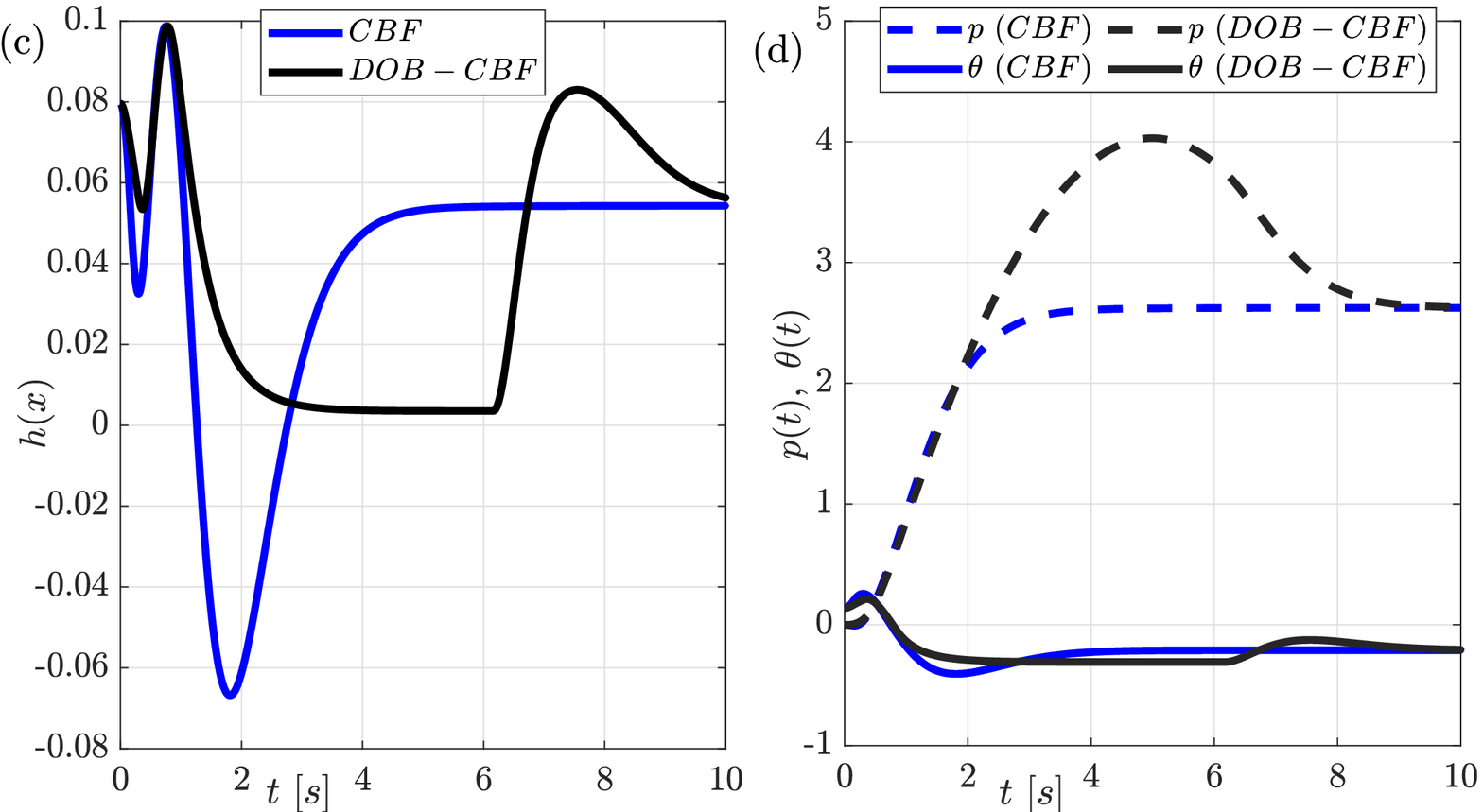}
  \end{subfigure}
  \caption{Segway example. (a) Planar Segway model on an inclined surface. (b) Estimated and actual effects of the disturbance and uncertainty on the time derivative of the control barrier function. (c) Control barrier function $h(x) \ [rad^2]$. (d) Trajectories of $p \ [m]$ and $\theta \ [rad]$. Simulations are performed with disturbance observer (DOB)-based robust CBF-QP (black) and nominal CBF-QP (blue).}
  \label{Segway}
\end{figure}
Fig.~\ref{Segway}-(c) and Fig.~\ref{Segway}-(d) show numerical simulation results where the Segway moves from $[0~ 0 ~0.138~ 0]^T$ to $[1~ 0 ~0.138~ 0]^T$ on a surface inclined by $\phi = 20^{\circ}$ in its state-space.
The planar Segway platform stays within the safe set with a disturbance observer-based approach while travelling on an inclined surface.
On the other hand, without disturbance observer, the Segway shows unsafe behaviour.
Fig.~\ref{Segway}-(b) shows that the proposed disturbance estimation approach appropriately estimate the actual effects of the disturbance and uncertainty on the time derivative of the CBF with the defined error bound.

\subsection{Conclusions and Future Work}
In this paper, we present a disturbance observer-based robust safe controller synthesis method in the presence of disturbance or uncertainty.
We first introduce a high-gain observer method to estimate the unmodelled dynamics of the CBF using only the safety constraint. 
Then, the estimated disturbance and associated error bound are utilized to construct a new robust CLF-CBF-QP.
We show the effectiveness of the method on the numerical simulations of an adaptive cruise control system and Segway with an external disturbance.
Our future work includes an extension of the proposed method to robust time-varying CBF approaches to consider the Signal Temporal Logic specifications.

\bibliographystyle{IEEEtran}
\bibliography{References.bib}

\end{document}